\documentclass[11pt]{article}
\usepackage[toc,page]{appendix}
\usepackage{hyperref}
\usepackage[T1]{fontenc}
\usepackage{amsfonts}
\usepackage[english]{babel}
\usepackage{a4wide,times}

\usepackage[latin1]{inputenc}
\usepackage{amssymb}
\usepackage{graphicx}
\usepackage{epsfig}
\usepackage{amsbsy}
\usepackage{verbatim}
\usepackage{color}
\usepackage{mathrsfs}
\usepackage{makeidx}
\usepackage{amsmath}
\usepackage{amsthm}
\usepackage{xcolor}

\usepackage{natbib}
\theoremstyle{plain}
\newtheorem{thm}{Theorem}[section]

\newtheorem{prop}[thm]{Proposition}

\theoremstyle{definition}

\usepackage{dsfont}

\theoremstyle{remark}
\newtheorem{Remark}{\bf Remark}[section]
\theoremstyle{remark}

\newtheorem{com*}{\bf Comment}
\usepackage{fancyhdr}

\makeatletter
\def \newequation#1#2{
 \@definecounter{#1}
 \@namedef{the#1}{\hbox{#2}}
 \@namedef{#1}{$$\refstepcounter{#1}}
 \@namedef{end#1}{
    \eqno \csname the#1\endcsname $$\global\@ignoretrue
    }
}

\newequation{haar}{(Haar)}

\newcommand{\E}{\mathds E}

\title{Long versus Short Time Scales: \\ the Rough Dilemma and Beyond}

\author{Matthieu Garcin\thanks{L\'eonard de Vinci P\^ole Universitaire, Research Center, 92916 Paris la D\'efense. Email: matthieu.garcin@devinci.fr.}
\and
Martino Grasselli \thanks{Department of Mathematics ``Tullio Levi Civita'',
University of Padova, via Trieste 63, 35121 Padova, Italy, and L\'eonard de Vinci P\^ole Universitaire, Research Center, 92916 Paris la D\'efense. Email: grassell@math.unipd.it. 
\newline 
Acknowledgments: We thank Elisa Al\`os, Fabienne Comte, Christa Cuchiero, Eva Flonner, Gilles Pag\`es, Andrea Pallavicini, and Mathieu Rosenbaum for useful comments on a preliminary version. We also thank the participants of the 2019 Quantitative Methods in Finance conference, Sydney for useful comments.
}
}

\date{2021}

\begin{document}
\maketitle
\begin{abstract}
Using a large dataset on major FX rates, we test the robustness of the rough fractional volatility model over different time scales, by including smoothing and measurement errors  into the analysis. Our findings lead to new stylized facts in the log-log plots of the second moments of realized variance increments
against lag which exhibit some convexity in addition to the roughness and stationarity of the volatility. The very low perceived Hurst exponents at small scales is consistent with the rough framework, while the higher perceived Hurst exponents for larger scales leads to a nonlinear behavior of the log-log plot that has not been described by models introduced so far. \end{abstract}



\section{Introduction}

It is well known that a constant volatility is not consistent with time series data
nor implied volatility surfaces. 
Several stochastic volatility models have been introduced in the last decades in order to reproduce the stylized facts of time series observed for both historical and implied volatility, see e.g. \cite{SS}, \cite{H}, \cite{BNS01}, \cite{SABR}, affine models like in \cite{DFS}, ARCH, GARCH, their non-parametric extensions \cite{GG}, and many others.
In a different (implied volatility) perspective, \cite{D} assumed that the volatility becomes a deterministic
function of time and of the current state of the asset, thus leading to the local volatility approach, which provides a theoretically perfect reproduction of the implied volatility surface.

All these approaches adopt the classic Brownian framework for
the noise of both the underlying and its volatility.
Motivated by an apparent presence of long memory in the volatility process, see e.g. \cite{L}, some researchers (see e.g. \cite{DGE}, \cite{Baillie}, \cite{BM}, \cite{AB}, \cite{BCD}...) modelled the log-volatility noise with a fractional Brownian motion, leading to the so called fractional stochastic volatility (FSV) model\footnote{A fractional Brownian motion (fBm) $B_H$ of Hurst exponent $H\in (0,1)$ is a Gaussian process with a non trivial covariance function, namely  a non Markovian process that allows for long or short memory, according to resp. $H>0.5$ or $H<0.5$. The case $H=0.5$ corresponds to the classic (Markovian) Brownian motion. 
}.

\cite{CR} suggested a model where the driving
fractional Brownian motion has Hurst parameter $H >0.5$, in
order to take into account for the stylized fact suggesting that the volatility is a  long memory process.
Their estimation procedure for the Hurst parameter was based on the method of \cite{GPH}, involving the slow decay of the autocorrelation function (which is supposed to be of power law with exponent less than one)  and it reveals to be problematic, as the asymptotic behaviour of the
covariance function cannot be directly computed without assuming a specific functional form. \\

Recently, a new paradigm for the volatility process has been introduced by  \cite{GJR}, who affirmed a universal phenomenon: volatility is rough and cannot be described by a SDE driven by a classic Brownian motion. In particular, they showed that the autocorrelation function of the volatility does not behave as a power law, at least at the time scales  ranging from one day till 2 months considered in  their observation, so that they disentangled the question about the long memory of the volatility from the asymptotic behaviour of the autocorrelation function. Then,   \cite{GJR} introduced a model, where the logarithm of the volatility is driven by a fractional Brownian motion (fBm) with a Hurst exponent that is empirically found to be  very low, thus leading to rough trajectories for the volatility. For this reason, the approach is referred to as the rough fractional stochastic volatility model (RFSV). Using  absolute moments estimation on a wide range of scales (from 1 day to approximatively 50 days), the Hurst exponent is found to be close to 0.14 both for the log-volatility of S\&P 500 and the NASDAQ, together with other major indexes. The series of volatilities used in the seminal paper \cite{GJR} come from the realized variance estimates from the Oxford-Man Institute of Quantitative Finance Realized Library
\footnote{\url{http://realized.oxford-man.ox.ac.uk/data/download}. The Oxford-Man Institute's Realized Library contains a selection of daily non-parametric estimates of volatility of financial assets, including realized variance and realized kernel estimates.}, between January 3, 2000 and March 31, 2014. In their approach, the daily squared volatility is estimated using the quadratic variations of the log-prices at a five-minute frequency (about 96 observations per day, in total about 3500 days). Another fundamental result in \cite{GJR} states that the estimation of the Hurst exponent $H$ is robust across time, scales and markets (equity indexes and FX). Before this empirical evidence of low Hurst exponents, based on series of historical volatilities, \cite{ALV} suggested that values of $H$ below $1/2$ should reproduce another stylized fact regarding implied volatilities, namely the short-date skew. \\

One may wonder if the results of \cite{GJR} depend on the  particular estimation procedure adopted. There are indeed different ways to compute the volatility proxies (see e.g. the Fourier-based estimators for the realized variance in \cite{CT} or the Parkinson estimator \cite{Parkinson} that one can  adopt when the realized variance is not available), or methods that even  circumvent the absolute moment estimation procedure, like the Whittle-type estimation methods used in \cite{FTW} to find directly the Hurst exponent. However, it is possible to show that all these methods lead to qualitatively similar results. We will show this and investigate daily time series on equity indexes using the Parkinson estimator in a forthcoming paper. \\

We thus observe two branches of the literature of fractional volatility leading to opposite conclusions:
\begin{itemize}
\item The traditional econometric approach, which studies the speed of the decay of the autocovariance function, concludes that there is long memory.
\item The rough volatility approach, in which the Hurst exponent is estimated from scaling properties of the series, uses thus estimators from literature of econophysics and of the statistics of stochastic processes. It concludes that there is no long memory.
\end{itemize}
In these two approaches, the tools are not the same, the range of scales analysed may differ, and the conclusions diverge. \\

A paper that somehow adopts both approaches is \cite{BLP}, where the authors present a two-factor stochastic volatility model which is rough at the short time scales, but it presents stationarity at longer time scales, according to the traditional approach that considers the autocovariance function. This mixing effect leads to an effective Hurst parameter varying on different observation time scales. In particular, the authors find a Hurst parameter for the S\&P 500 index ranging from few cents up to 0.2, as the observation time scale grows from one minute to few hundreds minutes. %, in line with our findings, where we get a similar result also for other equity indices \textcolor{blue}{(IN OUR PAPER, I THINK WE DO NOT HAVE SUCH RESULTS ON EQUITIES, SINCE WE FOCUS ON FX RATES)}. 
It is worth noticing that \cite{BLP} do consider just intra-day data: it would be interesting to see if their model is able to reproduce the stylized facts that we empirically observe also for longer time scales with daily data as we are going to describe. \\

Despite some technical issues arising from the fact that a volatility process driven by a fractional Brownian motion is not a semi-martingale (and the corresponding integrals require particular care), rough volatility models have been largely investigated in recent theoretical and empirical literature, to the point that today one can find more than one hundred papers on the subject\footnote{See e.g. the papers on the website \url{https://sites.google.com/site/roughvol/home}.}. Surprisingly,  the empirical investigation basically relies on the  dataset from the Oxford-Man Institute's Realized Library, 
which is indeed very useful, but  far to be the most complete.\\

The paper has two main contributions: 
the first is a modification of \cite{GJR}'s method taking errors
into account, and the second is to report some findings from FX data.
As volatility is unobserved, we include the estimation error into the analysis. In fact, what we observe is only noisy volatility.   We model two types of noise, coming from  measurement and smoothing (arising from the fact that we replace the unobserved spot variance with the integrated variance estimator). We quantify the impact of all these noises in the estimation of the Hurst exponent and we filter out their effects. What remains after the filtering  still deviates from the straight line behaviour predicted by the rough volatility model paradigm.
In particular, we extend previous empirical studies to a wider range of time scales, beyond the ones investigated in \cite{GJR}, in order to check whether data are consistent with the scaling properties predicted by fractional volatility models. The most striking difference with these models is a strong convexity of the log-log plot. In other words, the Hurst exponent perceived at small scales is lower than the Hurst exponent perceived at larger scales. We also include the mean reversion of the volatility in our investigation: when  we consider large time scales, we  expect to be able to observe the stationarity features that were precluded in the previous studies, since the time windows investigated so far were too small compared to the mean reversion frequency. 
In order to be able to observe such mean reversion effect, we obviously need long time series for the underlyings. 
Our results show that the presence of only one fBm is not enough in order to meet all the stylized facts, convexity and stationarity, empirically observed in our dataset. The main conclusion is that, by broadening the ranges of scales, we  highlight new stylized facts about the scaling of volatility processes that cannot be reproduced by the usual rough volatility model driven by a fractional Brownian motion. \\
\\

%This opens the door to an additional difficulty, because for equity indexes only daily data are available for time series starting before the 2000s, therefore we cannot use the realised volatility as a proxy of the unobserved daily volatility and we need another estimator.
% We suggest to adopt the Parkinson daily volatility estimator, see \cite{Parkinson}, which  reveals to give qualitatively the same results  of the  realized volatility estimator (in terms of the estimation of the Hurst exponent), when tested on the same dataset where intraday data are available.

%Endowed with a proxy for the daily volatility, we then repeat the empirical investigation of  \cite{GJR} and we confirm their findings for comparable time scales, while for larger scales our results are in contrast with the usual rough volatility paradigm, even by including  the mean reversion effect into the analysis.

The paper is organised as follows:  in Section \ref{Hurst}, we review  \cite{GJR} and other preceding studies
on the estimation of the Hurst parameter, such as \cite{FTW} that pointed out some source of spurious roughness.  Then, we propose a  modification that includes the measurement and smoothing errors, with some reasoning supported by a simulation
study. 
We re-examin the analysis of \cite{GJR} by including errors in Section \ref{sec:noise}, where we propose some ways to filter out the errors.
Finally, in Section \ref{sectionFX} we perform an empirical analysis
on FX for a large dataset and we provide evidence of new stylized facts that are not present in past literature.
Section \ref{Conclusion} concludes.

\section{Estimation of the Hurst exponent}\label{Hurst}

\subsection{Absolute moments estimation of the Hurst exponent}

Consider  the dynamics of a price process $S$ that evolves according to the SDE
\begin{equation}\label{SSDE}
dS_t=(.)dt+S_t \sigma_t dB_t, t \geq0,
\end{equation}
where $B$ is a standard Brownian motion defined in a probability space that satisfies the usual technical conditions. The process $S$ is assumed to be a semimartingale in order to avoid arbitrage opportunities. 

The volatility price process $\sigma$ is not directly observable, one can only deduce indirectly its properties through some observable proxies like the realized variance process, defined as 
\begin{equation}\label{realizedvariance}
\widehat \sigma^2_{\delta, t}=\sum_{(t-1)\delta\leq u\leq t\delta} \vert \Delta \log \overline S_u\vert^2,
\end{equation}
where $\overline S$ is a piecewise constant process which jumps at every sampling time of $S$ to the observed value of $S$ at the time.
If there is no measurement error and the sampling frequency goes to infinity we have that
\begin{equation*}
\widehat \sigma^2_{\delta, t}\rightarrow \int_{(t-1)\delta}^{ t\delta} \sigma^2_u du,
\end{equation*}
in probability, which justifies the choice of \eqref{realizedvariance} for a proxy of the realized variance process, as well as its square root for the daily spot volatility, in the case where $\delta$ corresponds to the length of one day.

In \cite{GJR}, the authors performed a linear regression in order to fit the empirical absolute moment of order $k$ of the log-volatility, defined as
\begin{equation}\label{kmoments}
\frac{1}{n}\sum_{t=1}^{n}{|\log\widehat \sigma_{\delta, t+\tau}-\log\widehat \sigma_{\delta, t}|^k}.
\end{equation}

They found a good fit, for different values of $k$, with the function
\begin{equation}\label{kmoments}
\log \frac{1}{n}\sum_{t=1}^{n}{|\log\widehat \sigma_{\delta, t+\tau}-\log\widehat \sigma_{\delta, t}|^k} \approx kH\log \tau + \eta_k,
\end{equation}
for a very small value of $H\approx 0.1$.
 
This special scaling property, together with some empirical stylized facts on the Gaussian nature of the log-variance (see e.g. \cite{ABDL01}), induced \cite{GJR} to assume a particular SDE for the log-volatility of the form
\begin{equation}\label{volaH}
d\log  \sigma^2_t = \eta dB^H_t,
\end{equation}
where $\eta$ is a constant and $B^H$ is a fractional Brownian motion with Hurst parameter $H$.

In fact, a fractional Brownian motion (fBm) of Hurst exponent $H\in(0,1)$ and scale parameter $\eta^2$ has stationary increments that satisfy
\begin{equation}\label{eq:absMom}
\E\left[|B^H_t-B^H_{t-\tau}|^k\right]=\frac{2^{k/2}\Gamma(\frac{k+1}{2})}{\Gamma(\frac{1}{2})}\eta^k \tau^{kH}, \qquad\tau, k \geq0,
\end{equation}
where $\Gamma(.)$ denotes the Gamma function, see \cite{Kolmogorov,MvN}.

%$$\E \left[B^H_tB^H_s\right]=\frac{\sigma^2}{2}(|t|^{2H}+|s|^{2H}-|t-s|^{2H}).$$

Turning things around, we  define the absolute empirical moment of order $k$ of the increments of a process $X$ (playing the role of the log-volatility process), in a time interval $[0,N]$ for a given scale $\tau$:\footnote{ In reality, we are using a version of $M_{k,\tau,N}(X)$ with overlapping increments. This allows us to slightly increase the convergence of this empirical absolute moment~\cite{LM}.}
\begin{equation}\label{kmoments}
M_{k,\tau,N}(X) = \frac{1}{\lfloor N/\tau\rfloor}\sum_{i=1}^{\lfloor N/\tau\rfloor}{|X_{i\tau}-X_{(i-1)\tau}|^k}.
\end{equation}
Using equation~\eqref{eq:absMom}, it follows that $\ln(M_{k,\tau,N}(X))$ is proportional to $H$ if $X$ is a fBm as  increments are stationary. This is the basis for estimators of Hurst exponents, see e.g. \cite{BCI,Garcin2017}. In particular, we can compute such empirical absolute moments for a great number of scales, and the estimator of $H$ is then $1/k$ times the slope of the regression of $\ln(M_{k,\tau,N}(X))$ on $\ln(\tau)$~\cite{Coeur2005}.
As a consequence, when plotting $\ln(M_{k,\tau,N}(X))$ as a function of $\ln(\tau)$ (also called log-log plot), we should get a straight line if $X$ is a fBm.

When restricting the analysis to small scales (say up to 2.5 months), one obtains such a straight line for log-log plot of realized volatilities, as first underlined by~\cite{GJR}. We reproduce this empirical observation in Figure~\ref{fig:NasSP}, with two stock indices and data from the Oxford-Man Institute's Realized Library.

\begin{figure}[htbp]
	\centering
		\includegraphics[width=0.55\textwidth]{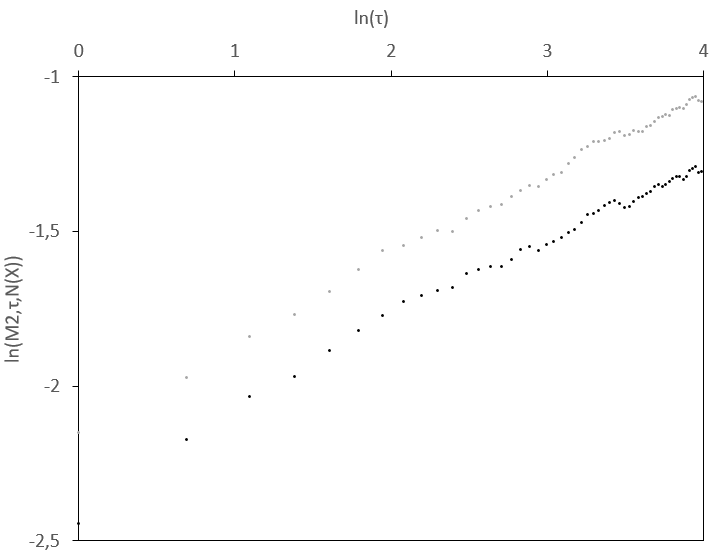}
\begin{minipage}{1.05\textwidth}\caption{Log-log plot, $\log (\tau)\mapsto\log (M_{2,\tau,1}(X))$, where $X$ is the log of realized volatility of two stock indices, NASDAQ index (black) and S\&P index (grey), between March 2000 and November 2018 sampled at a five-minute step. Time scales range between 1 day till 2.5 months.}
	\label{fig:NasSP}
\end{minipage}
\end{figure}

Note that the estimate $H\approx 0.1$ is in contrast with past literature on fractional volatility like \cite{CR}, where the long memory was associated to $H>0.5$. In particular, $H<0.5$ means that the volatility
path is rougher than semimartingales, and is consistent to a power law for the term structure of implied
volatility skew empirically observed in option markets, see \cite{GJR} and references therein.

\subsection{Impact of the mean-reversion at large scales}\label{mean_reversion}

An empirical stylized fact for the volatility is that it should be a stationary process, for both mathematical tractability and financial interpretability mostly at large times. In order to be consistent with a stationarity assumption, the rough volatility model in \cite{GJR} extends equation~\eqref{volaH} to the case where the log-volatility follows a fractional Ornstein-Uhlenbeck process (fOU, see \cite{CKM}) with a  very long reversion time scale, so that the effect of this mean reversion is invisible at the scales of the study (time scales considered are between one day and two months).
However, when dealing with scales longer than two months, we should expect a deviation from the linear behaviour in the log-log plot for  stationarized fBm processes (for example a fOU  or the inverse Lamperti transform of a fBm, see e.g. \cite{CKM,GarcinLamperti,SGKMBP}).
We illustrate the phenomenon in the simple  case where $k=2$, namely the second absolute moment for the fBm, that can be rewritten as follows: 

\begin{align*}
M_{2,\tau}(X)&= \E\left[|X_t-X_{t-\tau}|^2\right]\\
&=\E\left[|X_t|^2\right]+\E\left[|X_{t-\tau}|^2\right]-2Cov\left(X_t,X_{t-\tau}\right)
\end{align*}

Now, for large scales, the first two terms become similar and
independent of $\tau$. If the third term, the covariance, is a decreasing function wrt $\tau$ (decreasing to $0$), then $M_{2,\tau}(X)$ becomes an increasing function that flattens for large $\tau$. 
As a consequence, the function $\log (\tau)\rightarrow \log (M_{2,\tau}(X))$ behaves as a straight line for small values of $\tau$, but as the time scale increases, the slope decreases gradually until reaching zero and we observe a concave behaviour for large scales.
In conclusion, if we include the stationary volatility stylized fact into account, we should observe a potential decrease in the slope of the straight line in the log-log plot, or equivalently a smaller Hurst exponent for larger scales.

\subsection{Spurious roughness and noisy volatility}\label{sec:spurious}

It is well known that it is  possible to simulate  a spurious effect of roughness in the volatility paths just by playing with the drift of the volatility process, in particular by introducing a strong mean reversion effect, see for example Figure 2 in \cite{FTW} or \cite{Rogers19}. There is indeed another  easy way to generate spurious roughness, namely by adding noise into the observations, as also pointed out by \cite{FTW}. In fact, what we measure is not volatility but just noisy volatility.
In this subsection, we first provide a motivational example where we show that the presence of an additive noise may lead to a spurious roughness effect in a simulation study, where parameters are in line with empirical findings of the historical series that we shall consider in our empirical investigation. We then show how to quantify the impact of the noises in the estimation of the Hurst exponent in order to filter out the bias. 

%\subsubsection{Spurious roughness in a simple additive model}\label{sec:spurious}

We assume that the log-price $\log (S_t)$ follows an Ito process with stochastic variance  $\sigma^2_t$. Let $\widehat{\sigma}_t$ be the estimated volatility of day $t$, obtained using for example the realized volatility. These estimators provide us with a noisy version of the true and unobserved volatility, $\sigma_t$. We assume an additive model of measurement noise\footnote{Note that here we assume an additive model for the variance, not for the volatility. We can justify this choice by the additive nature of the variance, which makes it possible to apply the central limit theorem and to get an asymptotic distribution of the measurement error \cite{Rootzen80, JacodProtter98, BNS}.}: %,Meddahi
\begin{equation}\label{eq:VolBruit}
\widehat{\sigma}_t^2=\sigma_t^2+\varepsilon_t,
\end{equation}
where $t\in\{1,...,N\}$ and where the $\varepsilon_{1},...,\varepsilon_{N}$ are i.i.d. centered random variables.

Moreover, we consider the simple case where the variance of the log-prices follows a geometric Brownian motion:
\begin{equation}\label{eq:gBm}
\sigma_t^2=\sigma_0^2\exp\left(\beta B_t-\frac{1}{2}\beta^2t\right).
\end{equation}
Endowed with the estimation of the parameters $\sigma_0$ and $\beta$ for \eqref{eq:gBm}, we  simulate the model \eqref{eq:VolBruit} for different values of the standard deviation of the noise and it turns out that for certain values of the noise variance, a spurious roughness effect appears naturally, in line with \cite{FTW}: Figure \ref{noisyvol} shows a typical spurious rough simulated path for a  volatility satisfying \eqref{eq:gBm}.
\begin{figure}[htbp]
	\centering
\includegraphics[width=0.7\textwidth]{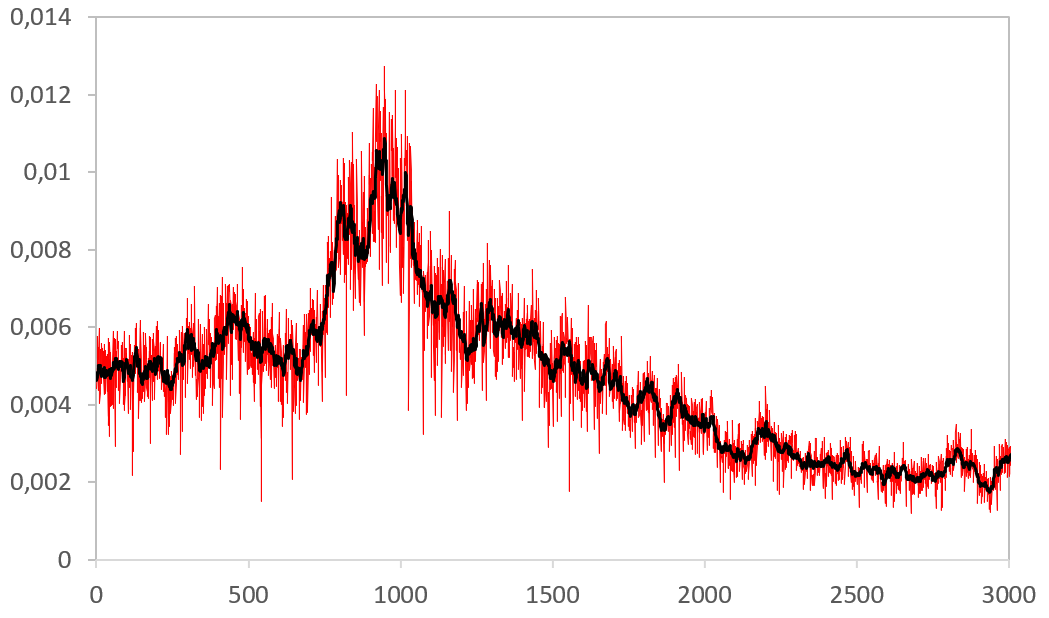}
\begin{minipage}{0.7\textwidth}\caption{Simulation of the noisy volatility  model~\eqref{eq:VolBruit}, in which the standard deviation of the noise is $25\%$ times the variance $\sigma^2_t$ of the log-price (the statistical framework exposed in Section~\ref{sec:TheoNoiseVolProx} justifies this parametric dependence to $\sigma^2_t$). The variance process of the price follows a geometric Brownian motion with the parameters estimated as for EUR/USD, i.e. $\sigma_0=4.62\times 10^{-3}$ and $\beta=3.8\%$.}
	\label{noisyvol}
\end{minipage}
\end{figure}

%
%The standard deviation of the noise conditionally to the volatility process is noted $\sqrt{\mathbb{V}ar\left[\varepsilon_t|\sigma_t\right]}=\alpha\sigma_t^2$, with $\alpha\geq 0$. In Figure~\ref{fig:SimulHvsNoise}, we see that spurious roughness ($H<0.2$) appears for $\alpha$ above $8\%$. In our framework (FX sampled every minute), $\alpha$ is only $3.7\%$, according to theoretical results exposed in the next Subsection~\ref{sec:TheoNoiseVolProx}. This factor $\alpha$ moves to $6.2\%$ for stocks sampled every minute, $13.9\%$ for stocks sampled every 5 minutes, $24.1\%$ for stocks sampled every 15 minutes.
%%, and $63.2\%$ for Parkinson approach.
%
%\begin{figure}[htbp]
%	\centering
%		\includegraphics[width=0.7\textwidth]{Illustrations/SimulHvsNoise.PNG}
%\begin{minipage}{0.7\textwidth}\caption{Estimated Hurst exponent at small scales (lower than 21 days, in black) and high scales (between 54 and 149 days, in grey) for various noise amplitude ($\alpha$, which is the standard deviation of the noise divided by the variance of the log-price), obtained by simulation of a noisy volatility. The variance process of the price follows a geometric Brownian motion with the parameters estimated as for EUR/USD, i.e. $\sigma_0=4.62\times 10^{-3}$ and $\xi=3.8\%$.}
%	\label{fig:SimulHvsNoise}
%\end{minipage}
%\end{figure}
%
This result strongly depends on the input volatility of the volatility (and on the model too, which is here limited to the geometric Brownian motion for the variance process). In particular, we used the same vol of vol for all the simulations, which is the vol of vol relevant for EUR/USD. With higher vol of vol, the impact of the noise in the estimation of the Hurst exponent is decreased, as one can see in Figure~\ref{fig:SimulHvsVoV_2}. 
\begin{figure}[htbp]
	\centering
		\includegraphics[width=0.6\textwidth]{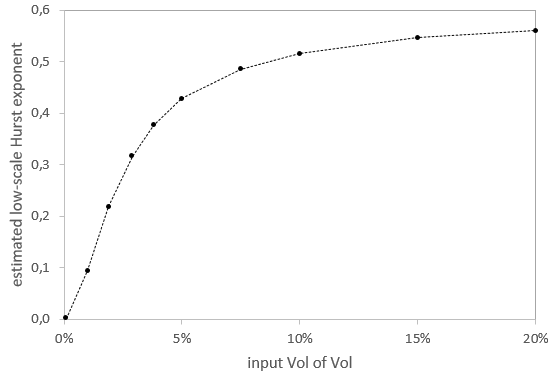}
\begin{minipage}{0.7\textwidth}\caption{Estimated Hurst exponent at small scales for various volatilities of the volatility, obtained by simulation of the prices and estimation of the realized volatility. The variance process of the price follows a geometric Brownian motion with the volatility parameter estimated as for EUR/USD, i.e. $\sigma_0=4.62\times 10^{-3}$, and $\beta$ ranging from $0.1\%$ till $20\%$. }
	\label{fig:SimulHvsVoV_2}
\end{minipage}
\end{figure}

In conclusion, a spurious roughness effect may appear in the presence of  a noise. It will be therefore  important to measure the standard deviation  of the noise in order to filter out the bias from the estimation of the Hurst exponent.

\subsection{Theoretical noise of the volatility proxy}\label{sec:TheoNoiseVolProx}

In this subsection we focus on the first source of bias in the estimates, namely the measurement error, coming from the  fact that the volatility process is not  observable. 
% We will focus on the two estimators considered so far, i.e. the realizd variance and the Parkinson estimators, for which  we quantify the corresponding  theoretical error.  
%
%\subsubsection{Realized volatility approach}\label{sec:TheorNoiseRV}

We can calculate the realized variance at a on-day scale for a given day, using $n$ log-returns:
$$\widehat{\sigma}^2_t(n)=\sum_{i=0}^{n-1}{\left[\log\left(\frac{S_{t-i/n}}{S_{t-(i+1)/n}}\right)\right]^2}.$$
The realized variance is an approximation of the integrated variance:
$$\bar{\sigma}^2_t=\int_{t-1}^{t}{\sigma^2_u du}.$$
If the log-price follows an Ito process, from e.g. \cite{Rootzen80, JacodProtter98, BNS}, we know the asymptotic distribution of the difference between the realized variance and the integrated variance, conditionally to the volatility process $(\sigma_u)_{u\in[t-1,t]}$:
$$\sqrt{n}\left(\widehat{\sigma}^2_t(n)-\bar{\sigma}^2_t\right)\overset{n\rightarrow\infty}{\longrightarrow} \mathcal N \left(0,2\int_{t-1}^t{\sigma^4_u du}\right).$$
That is, asymptotically,  $\widehat{\sigma}^2_t(n)-\bar{\sigma}^2_t$ is a centered Gaussian variable of variance $\frac{2}{n}\int_{t-1}^t{\sigma^4_u du}$. By neglecting the variations of $\sigma_t$ in a day, we approximate the variance of the $\varepsilon_t$ in equation~\eqref{eq:VolBruit} by $2\sigma_t^4/n$. For FX rates, sampled every minute, we have $n=1440$. %In average, for EUR/USD, as $\sigma_t$ is about $4.62\times 10^{-3}$ and $\sigma_t^2$ is $2.13\times 10^{-5}$ in the geometric Brownian motion model, the standard deviation of $\varepsilon_t$ is in average $7.95\times 10^{-7}$. 

\subsection{A second type of noise: the smoothing effect}

Motivated by the inaccuracy of the Hurst estimation, we now focus on the noise that comes from the fact that we try to estimate the Hurst exponent of the volatility process by applying estimation methods for the integrated volatility. Therefore, when we compute the variance of an increment of the realized volatility in the estimation of the Hurst exponent, for a given scale $\tau$, we are in fact computing the integrated variance of the increment of the volatility process for a scale varying between $\tau$ minus one day and $\tau$ plus one day.
This is referred to as to the smoothing error, namely the bias introduced by taking the integral instead of the spot (unobservable) argument, see also Appendix 3 in \cite{GJR}, where  methods and equations are similar to what we're going to show in this subsection.
In particular, in what follows, we show that this smoothing effect tends to overestimate the Hurst exponent. It may also explain why, in the Nasdaq log-log plots obtained by \cite{GJR}, the slope of straight line  is slightly higher for the lowest log-scales, what is illustrated in Figure~\ref{fig:Bruit_IV_logplot}.

\begin{figure}[htbp]
	\centering
		\includegraphics[width=0.7\textwidth]{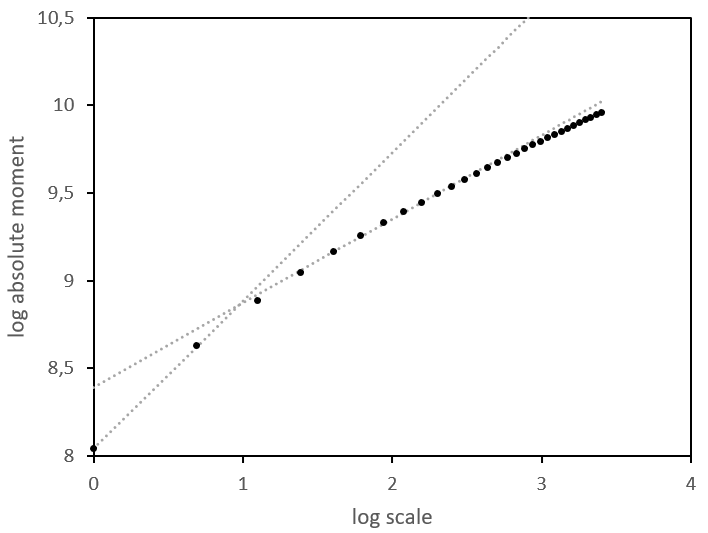}
\begin{minipage}{0.7\textwidth}\caption{Theoretical log second moment of the increments of $\widehat{\sigma}^2_{t,N,d}$, with $d=1$ and $N=100$, for an input Hurst exponent equal to 0.15. The grey dotted lines are the tangents for scales 1 and 2 (log-scales 0 and 0.7), and scales 5 and 10 (log-scales 1.6 and 2.3), leading respectively to estimated Hurst exponents of 0.42 and 0.24.}
	\label{fig:Bruit_IV_logplot}
\end{minipage}
\end{figure}

\begin{prop}\label{pro:smoothing}
Let the variance process $\sigma_t^2$ follow a fBm of Hurst exponent $H$ and variance $\xi^2$. Let the realized variance for $d\geq 1$ days and $N\geq 1$ observations per day be defined by: 
$$\widehat\sigma^2_{t,N,d}=\frac{d}{N}\sum_{i=0}^{N-1}{\sigma_{t-id/N}^2}.$$
Then, the variance of an increment of duration $\tau d$ of $\widehat\sigma^2_{t,N,d}$ is: 
$$\E\left[(\widehat\sigma^2_{t,N,d}-\widehat\sigma^2_{t-\tau d,N,d})^2\right] = \xi^2\frac{d^2}{N^2}\sum_{i=0}^{N-1}{\sum_{j=0}^{N-1}{\left(\left|\tau d + \frac{(j-i)d}{N}\right|^{2H}-\left|\frac{(j-i)d}{N}\right|^{2H}\right)}}.$$
For $d=1$, we also have the asymptotic expression:
\begin{equation}\label{eq:smoothingError}
\underset{N\rightarrow+\infty}{\lim} \E\left[(\widehat\sigma^2_{t,N,1}-\widehat\sigma^2_{t-\tau ,N,1})^2\right] = \xi^2\tau^{2H}f(\tau,H),
\end{equation}
where
\begin{equation}\label{eq:f_smoothingError}f(\tau,H)=\frac{\tau^{2}}{(2H+1)(2H+2)}\left(\left(1+\frac{1}{\tau}\right)^{2H+2}-2-2\left(\frac{1}{\tau}\right)^{2H+2}+\left(1-\frac{1}{\tau}\right)^{2H+2}\right).
\end{equation}
\end{prop}

\begin{proof}
We write first the straightforward decomposition:
$$\begin{array}{cl}
 & \E\left[(\widehat\sigma^2_{t,N,d}-\widehat\sigma^2_{t-\tau d,N,d})^2\right] \\
 = & \E\left[\widehat\sigma_{t,N,d}^4\right] + \E\left[\widehat\sigma_{t-\tau d,N,d}^4\right] - 2\E\left[\widehat\sigma^2_{t,N,d}\widehat\sigma^2_{t-\tau d,N,d}\right] \\
 = & \frac{d^2}{N^2}\sum_{i,j=0}^{N-1}{\left(\E\left[\sigma^2_{t-i\frac{d}{N}}\sigma^2_{t-j\frac{d}{N}}\right] + \E\left[\sigma^2_{t-\tau d-i\frac{d}{N}}\sigma^2_{t-\tau d-j\frac{d}{N}}\right] - 2\E\left[\sigma^2_{t-i\frac{d}{N}}\sigma^2_{t-\tau d-j\frac{d}{N}}\right]\right)}.
 \end{array}$$
Then, using the fact that $\E\left[\sigma^2_{t_i}\sigma^2_{t_j}\right]=\frac{\xi^2}{2}\left(|t_i|^{2H}+|t_j|^{2H}-|t_i-t_j|^{2H}\right)$, according to the fBm assumption, we get:
$$\begin{array}{ccl}
 \E\left[(\widehat\sigma^2_{t,N,d}-\widehat\sigma^2_{t-\tau d,N,d})^2\right] & = & \xi^2\frac{d^2}{2N^2}\sum_{i,j=0}^{N-1}{\left(\left[\left|t-i\frac{d}{N}\right|^{2H}+\left|t-j\frac{d}{N}\right|^{2H}-\left|\frac{(j-i)d}{N}\right|^{2H}\right] \right.} \\
 & & + \left[\left|t-\tau d-i\frac{d}{N}\right|^{2H}+\left|t-\tau d-j\frac{d}{N}\right|^{2H}-\left|\frac{(j-i)d}{N}\right|^{2H}\right] \\
 & & \left.- 2 \left[\left|t-i\frac{d}{N}\right|^{2H}+\left|t-\tau d-j\frac{d}{N}\right|^{2H}-\left|\tau d+\frac{(j-i)d}{N}\right|^{2H}\right]\right).
\end{array}$$
Noting that in the above equation $\sum_{i,j=0}^{N-1}{\left|t-\tau d-i\frac{d}{N}\right|^{2H}}=\sum_{i,j=0}^{N-1}{\left|t-\tau d-j\frac{d}{N}\right|^{2H}}$ and that $\sum_{i,j=0}^{N-1}{\left|t-i\frac{d}{N}\right|^{2H}}=\sum_{i,j=0}^{N-1}{\left|t-j\frac{d}{N}\right|^{2H}}$, we finally obtain:
$$\E\left[(\widehat\sigma^2_{t,N,d}-\widehat\sigma^2_{t-\tau d,N,d})^2\right] = \xi^2\frac{d^2}{N^2}\sum_{i=0}^{N-1}{\sum_{j=0}^{N-1}{\left(\left|\tau d + \frac{(j-i)d}{N}\right|^{2H}-\left|\frac{(j-i)d}{N}\right|^{2H}\right)}}.$$
From this result, we conclude for the asymptotic case:
$$\begin{array}{ccl}
\underset{N\rightarrow+\infty}{\lim} \E\left[(\widehat\sigma^2_{t,N,1}-\widehat\sigma^2_{t-\tau ,N,1})^2\right] & = & \xi^2\int_0^1{\int_0^1{\left(\left|\tau +v-u\right|^{2H}-\left|v-u\right|^{2H}\right)du}dv} \\
 & = & \xi^2\tau^{2H}f(\tau,H).
\end{array}$$
\end{proof}

The asymptotic version of the realized variance $\widehat\sigma_.$ is the integrated variance $\bar\sigma_.$. In Proposition~\ref{pro:smoothing}, we get an expression for the variance of its increments which is consistent with the one provided by~\cite{GJR} and which is more concise than the equivalent expression obtained for the realized variance.

As one can see in Figure~\ref{fig:Bruit_IV_Hurst}, though the bias is limited for higher values of $H$, it is very significant for lower Hurst exponents.

\begin{figure}[htbp]
	\centering
		\includegraphics[width=0.6\textwidth]{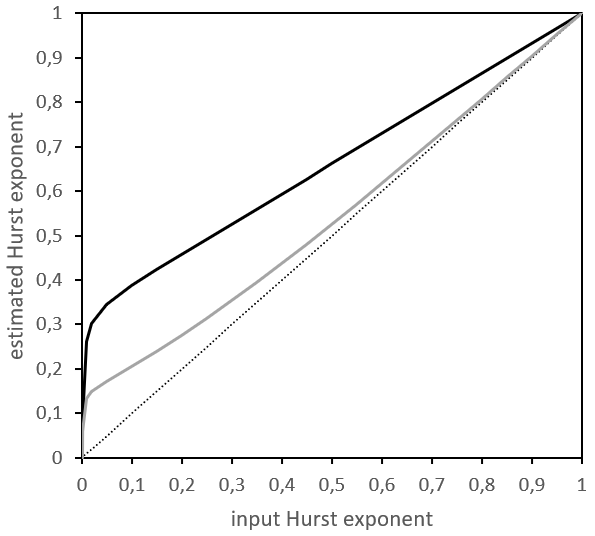}
\begin{minipage}{0.7\textwidth}\caption{Theoretically estimated Hurst exponent for $\widehat{\sigma}^2_{t,N,d}$, with $d=1$ and $N=100$, for various values of input Hurst exponents. The black line is the estimation on scales 1 and 2 (log-scales 0 and 0.7), the grey one on scales 5 and 10 (log-scales 1.6 and 2.3). The dotted line is the identity.}
	\label{fig:Bruit_IV_Hurst}
\end{minipage}
\end{figure}

%Similarly, Parkinson volatility is subject to two kinds of noise. The first one is the measurement noise and is not be neglected as this proxy is based on only two observations. On the contrary, this low number of observations makes the microstructure noise almost nonexistent. The second kind of noise is similar to the smoothing error: in fact, as the real volatility is not constant over one day, Parkinson volatility does not proxy a spot volatility. Depending on when the maximum and the minimum prices are reached in a day, the autocovariance structure of this proxy may differ from the one of the spot volatility. Unfortunately, we are not able to quantify accurately this noise. Simulations, presented in Section~\ref{sec:FilterParkinson}, suggest however that it has a lower amplitude than the measurement noise.

\begin{Remark}
There is also a third type of error, namely the microstructure noise, coming from decimalization, absence of distinction between bid and ask, price formation and other issues typicically faced in intraday data (see e.g. \cite{RR1, RR2} and references therein). This microstructure noise affects the price at each observation in a roughly similar manner as the measurement noise already analysed. It tends to overestimate the variance process. In order to limit the impact of the microstructure noise, one can consider higher time steps for estimating the integrated variance. In line with the empirical literature, a 5-minute or even a 1-minute frequency turns out to limit the impact of the microstructure noise in the FX market \citep{BKZ,LA,KM}, so we will neglect this error in our empirical analysis. 

\end{Remark}

\section{Filtering the noises}\label{sec:noise}

%We have inventoried two sources of noise:  a measurement noise and a smoothing error. 
%The estimator of the integrated variance is affected by the choice of the time step of the observations. In particular,  there is a bias n underestimation of the slope in the log-log plot, i.e. a smaller Hurst exponent. This is visible in Figure~\ref{fig:Filtre}, in which the time step for the realized variance was increased from 1 minute to 40 minutes in the historical series of EUR/USD.
In this section, we are only working with moments of order 2, with an increment duration $\tau\in\mathbb N$ such that $1\leq\tau\ll N$:
\begin{equation}\label{2moments}
M_{2,\tau}(X)=\frac{1}{(N-\tau+1)}\sum_{i=\tau}^{N}{\vert X_{i}- X_{i-\tau}\vert^2}.
\end{equation}

We want to filter the measurement noise and the smoothing error. More precisely, for the measurement noise we simply correct the bias in the estimation of all the absolute moments in the log-log plot for the variance, and for the smoothing error we correct the multiplicative error in the moment appearing in equation~\eqref{eq:smoothingError}. To put it simply, we are looking for the moments of increments of $\sigma^2_t$, but as we cannot observe $\sigma^2_t$ we work instead with $\bar{\sigma}^2_t$, thus introducing the smoothing error. In addition, $\bar{\sigma}^2_t$ is approximated by $\widehat{\sigma}^2_t$, whose difference with $\bar{\sigma}^2_t$ is the measurement noise.

The way the noise is to be filtered strongly depends on the model one assumes for the volatility dynamic. In the next two subsections, we present two filtering methods based on two competing models. The first one is consistent with the RFSV approach. The second model uses the same underlying dynamic, a fBm, but applies it to the variance process instead of the volatility process. This approach is more consistent with dynamics inspired by the ARCH process, which are traditionally invoked in econometrics. %It is worth noting that in the next section we shall show that the filtered log-log plots still present some convexity. This confirms our findings regarding new stylized facts. 

\subsection{The log-volatility as a fBm}

The first model we consider if the RFSV model~\cite{GJR}: $\sigma_t=\sigma\exp(\xi B^H_t)$, where $B^H_t$ is a fBm of Hurst exponent $H$ and $\sigma,\xi>0$. With this assumption, the theoretical log-log plot may differ from a straight line because of measurement noise and smoothing error. We present a method to filter these noises from the log-log plot. If the RFSV model depicted all the scaling features of the spot volatility series, the filtered plot should result in a straight line.

%\subsubsection{Realized volatility approach}

The log-log plot of the log-volatility should be based on $M_{2,\tau}(\log(\sigma))$. However, we only observe $M_{2,\tau}(\log(\widehat{\sigma}))$. Thanks to the successive approximations explained below, we get the following relation between $M_{2,\tau}(\log(\sigma))$ and $M_{2,\tau}(\log(\widehat{\sigma}))$:
\begin{equation}\label{eq:filtre}
\begin{array}{ccl}
M_{2,\tau}(\log(\sigma_.)) & = & \frac{1}{4} M_{2,\tau}(\log(\sigma^2_.)) \\
 & \approx & \frac{1}{4\sigma_0^4}M_{2,\tau}(\sigma^2_.) \\
 & \approx & \frac{1}{4\sigma_0^4}f(\tau,H)^{-1}M_{2,\tau}(\bar{\sigma}^2_.) \\
 & \approx & \frac{1}{4}f(\tau,H)^{-1}M_{2,\tau}(\log(\bar{\sigma}^2_.))\\
 & \approx & \frac{1}{4}f(\tau,H)^{-1}\left[M_{2,\tau}(\log(\widehat{\sigma}^2_.))-\frac{4}{n}\right] \\
 & = & f(\tau,H)^{-1}\left[M_{2,\tau}(\log(\widehat{\sigma}_.))-\frac{1}{n}\right] ,
\end{array}
\end{equation}
where $\sigma_0\in\mathbb R$ and $f$ is defined in equation~\eqref{eq:f_smoothingError}. 
\begin{itemize}
\item The approximation between the first and the second line is based on a first-order Taylor expansion of the logarithm, around an arbitrary value $\sigma_0^2$: 
$$\begin{array}{ccl}
M_{2,\tau}(\log(\sigma^2_.)) & = & \E[(\log(\sigma^2_{.+\tau})-\log(\sigma^2_{.}))^2] \\
 & \approx & \E\left[\left(\log(\sigma_0^2)+\frac{\sigma^2_{.+\tau}-\sigma_0^2}{\sigma_0^2}-\log(\sigma_0^2)-\frac{\sigma^2_{.}-\sigma_0^2}{\sigma_0^2}\right)^2\right] \\
 & = & \frac{1}{\sigma_0^4}\E\left[\left(\sigma^2_{.+\tau}-\sigma^2_{.}\right)^2\right] \\
 & = & \frac{1}{\sigma_0^4}M_{2,\tau}(\sigma^2_.). 
\end{array}$$
In addition to the error left by the Taylor expansion, another source of error may appear in the equations above and in all this section regarding the difference between the empirical moments, $M_{2,\tau}$, and the theoretical ones. However, for long time series as ours, the difference is small with respect to other approximations and is equal to zero in average.
\item The approximation between the second and the third line is based on equation~\eqref{eq:smoothingError}, that is $M_{2,\tau}(\bar{\sigma}^2_.)=M_{2,\tau}(\sigma^2_.)f(\tau,H)$, and thus filters the smoothing error, assuming that $\sigma^2$ follows a fBm of Hurst exponent $H$. In fact, the assumption that $\log (\sigma)$ follows a fBm, which is the assumption in line with the  rough volatility model of \cite{GJR}, leads to the same approximation thanks to the first-order Taylor expansion of the logarithm introduced above, as $M_{2,\tau}(\log (\sigma_.))\approx \frac{1}{4\sigma_0^4}M_{2,\tau}(\sigma^2_.)$, which can be seen by following the same argument as in Proposition B.1 of \cite{FTW} who quantified the difference
between the log of integrated variance and the integral of log variance. 
\item The approximation between the third and the fourth line is based on the same kind of Taylor expansion than between the first and the second line.
\item The approximation between the fourth and the fifth line exploits the measurement noise between $\widehat{\sigma}^2$ and $\bar{\sigma}^2$. We have seen that if the log-price follows an Ito process, then $\widehat{\sigma}^2_t=\bar{\sigma}^2_t+\varepsilon_t$ where $\varepsilon_t$ is asymptotically a centered Gaussian variable of variance $\frac{2}{n}\int_{t-1}^{t}{\sigma_u^4du}\approx \frac{2}{n}\bar{\sigma}_t^4$. Thus, as an approximation, we consider that $\widehat{\sigma}^2_t\approx \bar{\sigma}^2_t(1+\alpha_t)$, with $\alpha_t\sim\mathcal N(0,2/n)$. According to Theorem H.1 of \cite{FTW}, the $\alpha_t$ are independent variables, also independent of the $\bar{\sigma}^2_t$, then:
\begin{equation}\label{eq:MomRVvsIV}
\begin{array}{ccl}
M_{2,\tau}(\log (\widehat{\sigma}^2_.)) & \approx & \E[(\log (\bar{\sigma}^2_{.+\tau})-\log (\bar{\sigma}^2_{.})+\log (1+\alpha_{.+\tau})-\log (1+\alpha_{.}))^2] \\
 & = & M_{2,\tau}(\log (\bar{\sigma}^2_.)) + \E[(\log (1+\alpha_{.+\tau})-\log (1+\alpha_{.}))^2] \\
 & = & M_{2,\tau}(\log (\bar{\sigma}^2_.)) + 2\E[(\log (1+\alpha_{.}))^2] - 2(\E[\log (1+\alpha_{.})])^2 \\
 & = & M_{2,\tau}(\log (\bar{\sigma}^2_.)) + 2\mathbb{V}ar (\log (1+\alpha_{.})) \\
 & \approx & M_{2,\tau}(\log (\bar{\sigma}^2_.)) + 2\mathbb{V}ar (\alpha_{.}) , 
\end{array}
\end{equation}
where we go from the first to the second line thanks to the fact that $\alpha_{.+\tau}$ and $\alpha_{.}$ are identically distributed, then to the third line because they are iid and to the last line by a Taylor expansion for the second moment of the logarithm of a random variable.
\end{itemize}

In equation~\eqref{eq:filtre}, we filter successively the measurement noise by translating each absolute moment by a value of $1/n$. Then, we filter the smoothing error by dividing the result by $f(\tau,H)$, with a properly chosen $H$, consistent with the rough framework.

We  show on a theoretical example, illustrated in Figure \ref{fig:FiltrePark}, that,  despite all the approximations cited above, the filtering method we propose is fairly accurate.

\begin{figure}[htbp]
	\centering
		\includegraphics[width=0.7\textwidth]{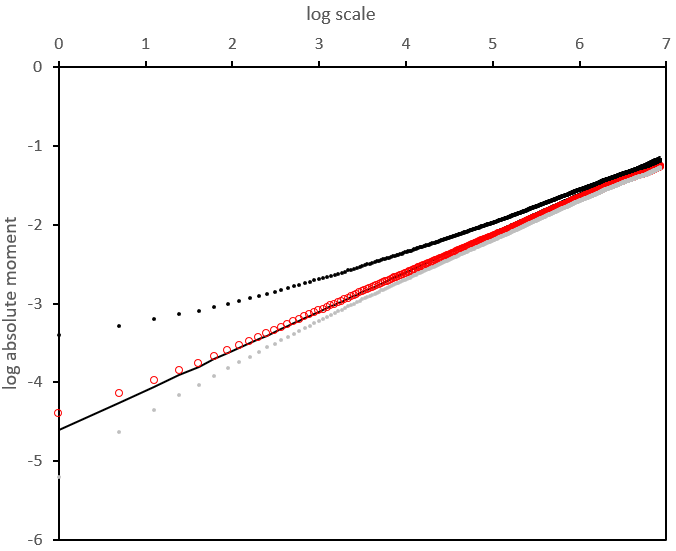}
\begin{minipage}{0.7\textwidth}\caption{Average of $\log (\tau)\mapsto\log (M_{2,\tau}(X))$, where $X$ is the log of the realized volatility (with $n=36$ observations each day, thus corresponding to a $40$-minute time step; as a robustness check, we have led the same analysis for other values of $n$ and have obtained similar results) of simulated prices with a volatility following an RFSV model, with $H=0.25$ and $\xi=10\%$. The black dots are the unfiltered observations and the straight line is the theoretical log-log plot for the RFSV model. The grey dots are the observations filtered by removing $1/n$ from each moment (filtering of the measurement noise), the red ones by filtering both the measurement noise and the smoothing error by the means of equation~\eqref{eq:filtre}.
 }
	\label{fig:FiltrePark}
\end{minipage}
\end{figure}

\subsection{The log-variance as a fBm}\label{logVariance}

If, instead of the RFSV model assumed in the subsection above, the log-variance follows a fBm, the approximations are then transformed into the following, when using the realized variance proxy:
$$\begin{array}{ccl}
M_{2,\tau}(\log (\sigma^2_.))  & \approx & \frac{1}{\sigma_0^4}M_{2,\tau}(\sigma^2_.) \\
 & \approx & \frac{1}{\sigma_0^4}f(\tau,H)^{-1}M_{2,\tau}(\bar{\sigma}^2_.) \\
 & \approx & f(\tau,H)^{-1}M_{2,\tau}(\log (\bar{\sigma}^2_.))\\
 & \approx & f(\tau,H)^{-1}\left[M_{2,\tau}(\log (\widehat{\sigma}^2_.))-\frac{4}{n}\right].
\end{array}$$

The main difference with the model in which the log-volatility is an fBm is a factor 4 in the filter of the measurement noise. %appears, and in this case the filter of the smoothing error is more accurate. \textcolor{red}{[ENOUGH?..]}
%In the Parkinson approach, the variance $\sigma_t^2$ is now estimated by $d_t^2/4\log(2)$, as exposed in Section~\ref{sec:ParkinsonNoise}, whose measurement noise has a conditional variance equal to $V\sigma_t^4$, with $V=\left(\frac{9\zeta(3)}{16\left(\log(2)\right)^2}-1\right)$ according to equation~\eqref{eq:ParkVarNoise}. The extension of the filtering method exposed in Section~\ref{sec:FilterParkinson} to the variance case thus consists in removing $2V$ from each absolute moment of the increments of the Parkinson variances series. Again, this method neglects the smoothing error.
%We have conducted the same analysis as above on FX rates, by replacing proxies of volatility by proxies of variance. The filtered version is very similar, with still a clear convexity: the slope of the log-log plot increases with the time scale. We have not displayed the log-log plots of empirical variance processes in this paper since they are very close to the ones we obtained for the volatility. It supplements our findings of the previous subsection: neither the volatility nor the variance seems to follow a fBm, according to both raw and filtered log-log plots.

\section{Empirical results for exchange rates}\label{sectionFX}

In this section we implement the empirical absolute moment regressions for volatilities of the most liquid exchange rates. We base our analysis on a real data set from Interactive Brokers. It gathers high-frequency rates between the 18th December 2006 and the 19th June 2019, sampled at a one-minute step, for a total of more than 4.6 millions data for each of the following pairs: EUR/USD, EUR/GBP, EUR/JPY, EUR/CAD, EUR/AUD, GBP/USD, GBP/JPY, USD/JPY, AUD/USD, and AUD/JPY. From these rates, we estimate daily volatilities, using the 1,440 observations in each day. We thus have for each series 3,206 consecutive observations of daily volatility that will be estimated by the realized volatility, defined as the square root of the average quadratic one-minute log-variation. We display the  results in Figure~\ref{fig:HurstFX} and Table~\ref{tab:HurstFX}.
\begin{figure}[htbp]
	\centering
		\includegraphics[width=0.38\textwidth]{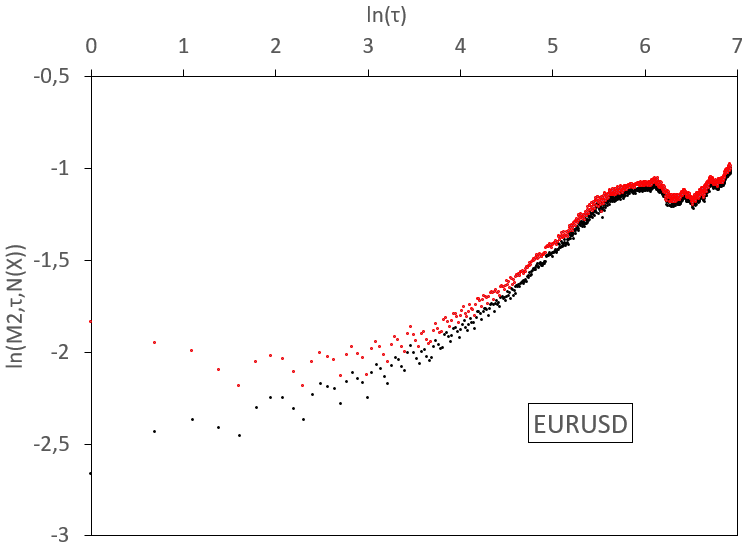} 
		\includegraphics[width=0.38\textwidth]{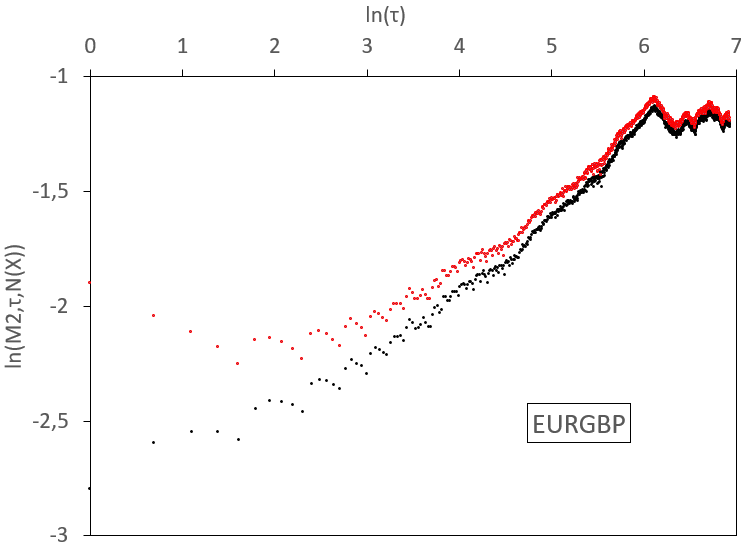}
		\includegraphics[width=0.38\textwidth]{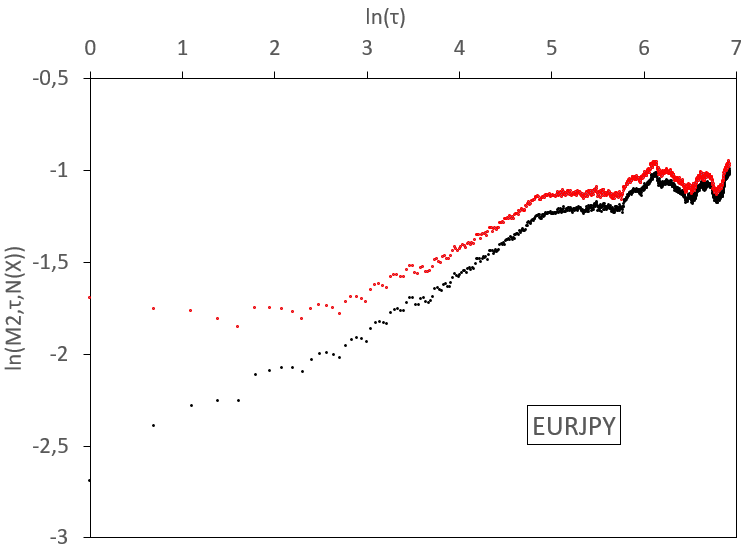}
		\includegraphics[width=0.38\textwidth]{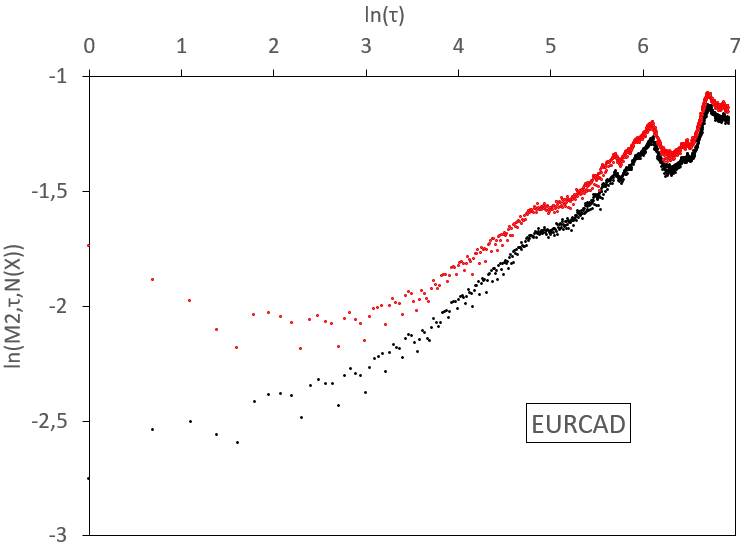}
		\includegraphics[width=0.38\textwidth]{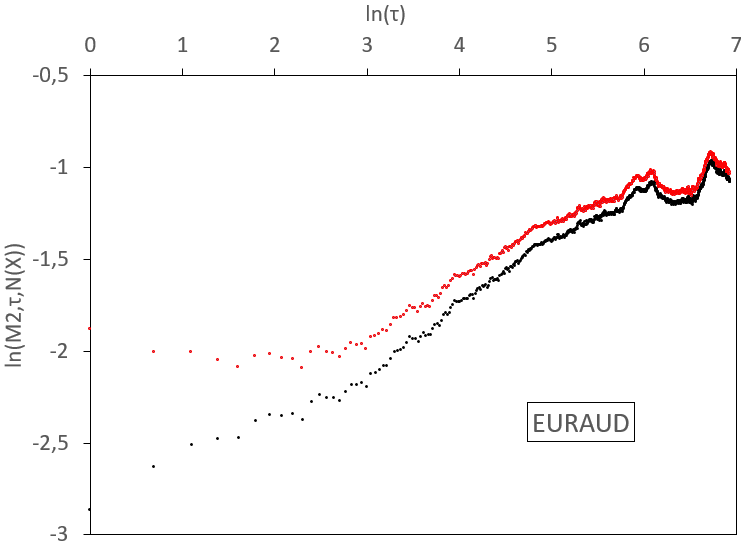}
		\includegraphics[width=0.38\textwidth]{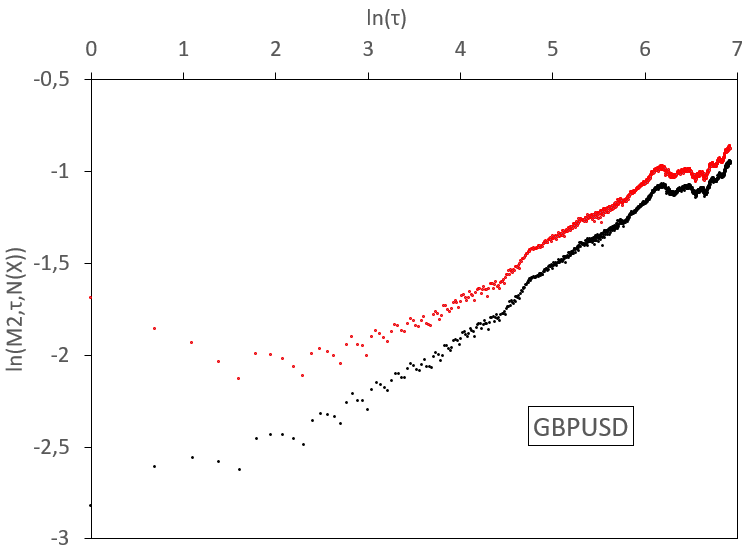}
		\includegraphics[width=0.38\textwidth]{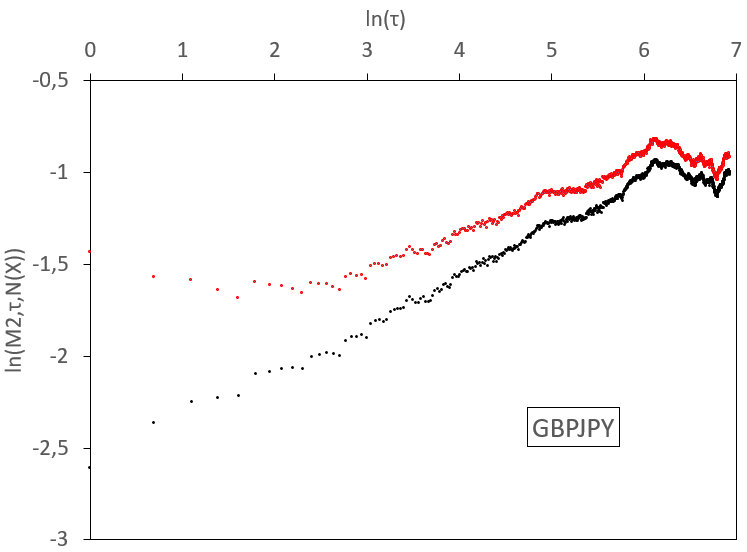}
		\includegraphics[width=0.38\textwidth]{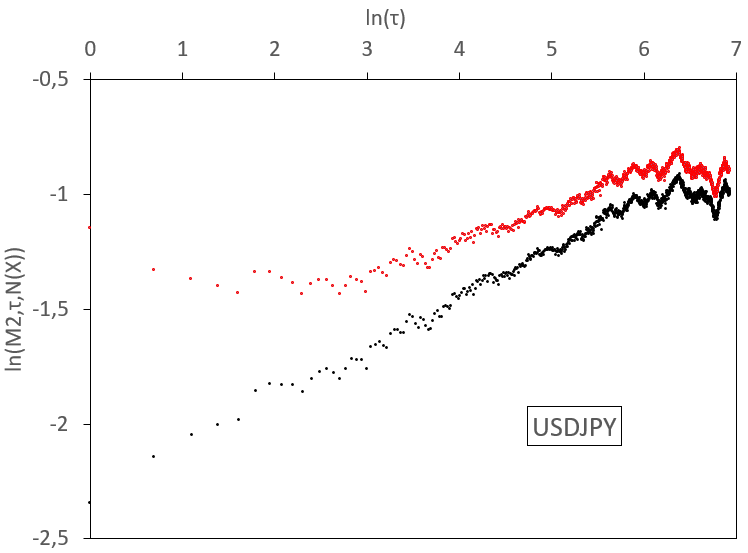}
		\includegraphics[width=0.38\textwidth]{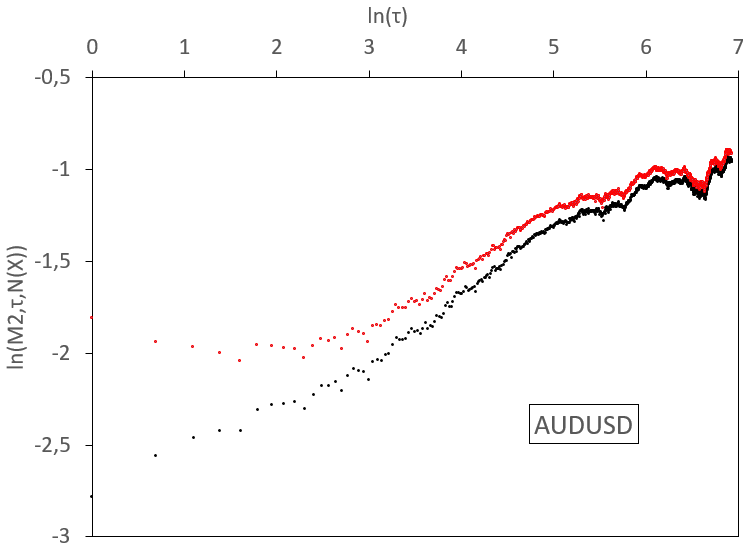}
		\includegraphics[width=0.38\textwidth]{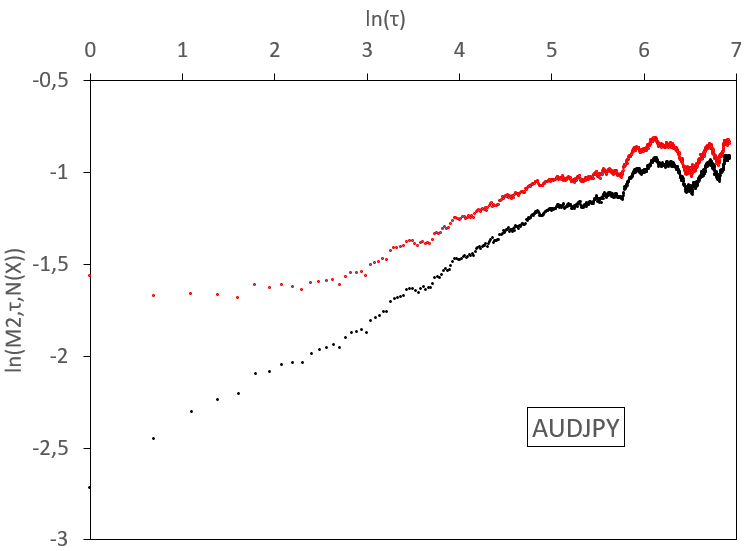}
\begin{minipage}{1.05\textwidth}\caption{Log-log plot (in black), $\log (\tau)\mapsto\log (M_{2,\tau,1}(X))$, where $X$ is the log of realized volatility of various FX rates between December 2006 and June 2019 sampled at a one-minute step, for a total of 4.6 millions data  for each of the following pairs: EUR/USD, EUR/GBP, EUR/JPY, EUR/CAD, EUR/AUD, GBP/USD, GBP/JPY, USD/JPY, AUD/USD, and AUD/JPY. Time scales range between 1 day till 3 years. The log-log plot in red is obtained by applying the denoising filter. }
	\label{fig:HurstFX}
\end{minipage}
\end{figure}

In Figure~\ref{fig:HurstFX}, we observe various regimes for the volatility process, depending on the time scale.  In general, for log-scales between 0 and 3 (that is to say between 1 day and 3 weeks), a linear regression of small slope holds for empirical absolute moments. Between 3 and a higher abscissa which depends on the rate considered (so between 3 weeks and a time horizon $\mathcal H$, which is located between 4.5 months and 12 months), the slope steepens. Above this long time horizon $\mathcal H$, the slope flattens and, at scales which depend again on the sample, aberrant oscillations appear, due to the limited size of the sample. Therefore, we observe a threefold behaviour of the dynamics: for small time scales (1 day till 3 weeks)  we get a small Hurst exponent, at higher scales (3 weeks till $\mathcal H$)  we have a higher Hurst exponent, and for very large scales (above $\mathcal H$)  we observe stationarity.

\begin{table}[htbp]
\centering
\begin{tabular}{l||c|c||c|c}
\hline
\hline
%\textbf{Type of volatility} & \multicolumn{2}{c||}{\textbf{Realized volatility}} & \multicolumn{2}{c}
%{\textbf{Parkinson's volatility}} \\
\hline
 \textbf{Scale} & \textbf{Small scales} &  \textbf{Large scales} \\
\hline
\hline
\textbf{EUR/USD} & 0.075 & 0.244 \\
\textbf{EUR/GBP}  & 0.086 & 0.217  \\
\textbf{EUR/JPY}  & 0.117 & 0.190 \\
\textbf{EUR/CAD}  & 0.068 & 0.184 \\
\textbf{EUR/AUD}  & 0.106 & 0.192  \\
\textbf{GBP/USD}  & 0.093 & 0.157  \\
\textbf{GBP/JPY}  & 0.112 & 0.148 \\
\textbf{USD/JPY}  & 0.098 & 0.144 \\
\textbf{AUD/USD}  & 0.108 & 0.196  \\
\textbf{AUD/JPY}  & 0.135 & 0.152 \\
\hline
\hline
\end{tabular}
\begin{minipage}{0.7\textwidth}\caption{Perceived Hurst exponent before denoising, estimated as half the slope of linear regressions for the plots presented in Figure~\ref{fig:HurstFX} restricted a time horizon lower than $\mathcal H$. Small scales refer to the range from 1 day till 3 weeks, large scales refer to a range depending on the sample considered: from 2 months till 4.5 months for EUR/JPY, EUR/CAD, EUR/AUD, GBP/JPY, and AUD/JPY, from 2 till 5 months for AUD/USD, from 3 till 8 months for EUR/USD, from 5 till 9 months for USD/JPY, and from 5 till 12 months for EUR/GBP and GBP/USD.}
\label{tab:HurstFX}
\end{minipage}
\end{table}

Now we turn our attention to the measurement noise and the smoothing error. We apply the methodology described in the previous section and we filter these noises from the observations. 
We display the results for the EUR/USD time series in Figure~\ref{fig:Filtre}, with the successive application of filters. 

%We observe that the filtering of the sole microstructure noise, by increasing the time step considered in the computation of the realized variance, leads to a flattening of the whole curve (from the black curve to the dark grey curve). This is counterintuitive, as the filtering of noise should lead to a higher estimated Hurst exponent, that is to a higher slope of the log-log plot. However, when we increased the time step to reduce the microstructure noise, other phenomena also come into play. For instance, the increase of the time step reduced the number of observations in a day and therefore we increased the measurement noise. 
The filtering of the measurement noise steepens the log-log plot (grey curve). With a 5-minute time step in the computation of realized variance, this effect is slight. Like for the initial curve, we observe convexity when the measurement noise is filtered: higher estimated Hurst exponent for low-frequency increments than for high-frequency increments of volatility. The second filter is about the smoothing error. It accentuates the convexity (red curve). We conclude that the filtered log-log plot still shows convexity. The fBm assumption, for which we should have the same slope of the log-log plot for all scales, is not enough to explain the observations.

\begin{figure}[htbp]
	\centering
		\includegraphics[width=0.7\textwidth]{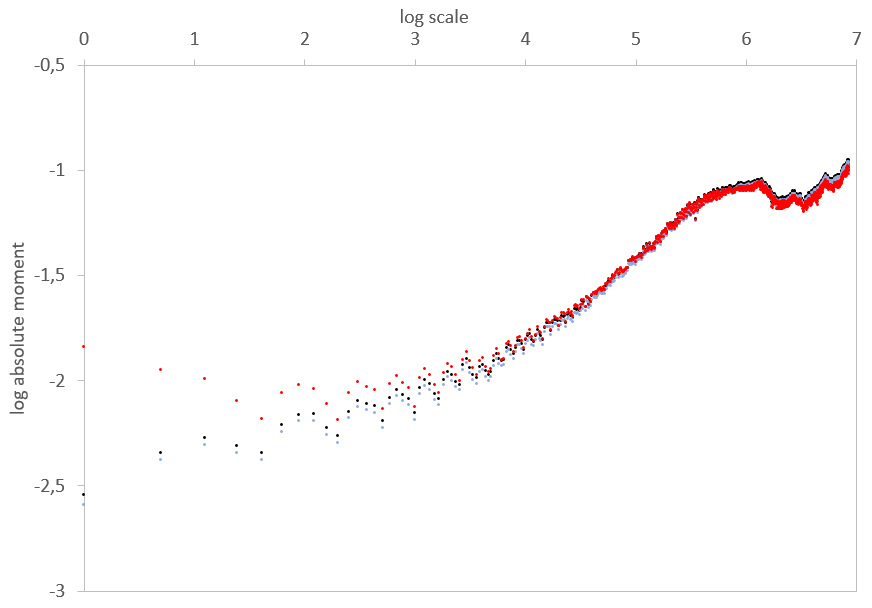}
\begin{minipage}{0.7\textwidth}\caption{Filtering noise for EUR/USD. 
%Microstructure noise has been filtered by increasing the time step for the realized variance from 1 minute (log-log plot in black) to 40 minutes (log-log plot in dark grey). 
Log-log plot for the the realized variance with 5-minute time step (in black). The measurement noise has been filtered by translating each absolute moment by a value of $1/n$ as in Equation \eqref{eq:filtre} (log-log plot in grey). Smoothing error has been additionally filtered by dividing the result by $f(\tau,H)$, with  $H$ chosen to be in line with the estimates (the results are indeed robust wrt the choice of $H$, consistent with the rough framework (log-log plot in red). Time scales range between 1 day till 3 years.
 }
	\label{fig:Filtre}
\end{minipage}
\end{figure}

We display similar results for all the studied FX rates in Figure~\ref{fig:HurstFX}. We consider a 1-minute time step in the computation of the realized variance. For FX rates, this amounts to 1440 observations each day and the corresponding measurement noise is almost invisible. We therefore do not decompose the noise filtering in this figure and simply display the raw (black) and fully filtered (red) curves. Once again, we observe convexity in the raw and filtered log-log plots.

These results deserve some comments, since the filtering relies on some approximations:
\begin{itemize}
\item Taylor expansions of the logarithm.
\item Identification of the empirical moments to the theoretical ones.
\item For the filtering of the smoothing error, we assume that the distortion of the log-log plot is the one consistent with a model in which the variance follows an fBm, rather than the rough volatility model that assumes instead that the log-volatility is an fBm. Nevertheless, the smoothing error tends to make the log-plot more concave because the increments of duration $\tau$ of the realized variance mitigates increments of the spot variance of duration $\tau$ with increments of shorter duration, for which the covariance is much stronger. However, we still observe that, even without this approximative filtering, the log-plot is convex in our empirical findings, see Figure \ref{fig:Filtre}. On top of that, we have conducted the same analysis as above on FX rates, by replacing proxies of volatility by proxies of variance. The filtered version is very similar, with still a clear convexity: the slope of the log-log plot increases with the time scale. We have not displayed the log-log plots of empirical variance processes in this paper since they are very close to the ones we obtained for the volatility. It supplements our findings of the previous subsection: neither the volatility nor the variance seems to follow a fBm, according to both raw and filtered log-log plots. 
\item The filtering function $f$ for the smoothing error depends on an $H$ parameter. But the value of the true Hurst exponent is unknown and should be the output of the filtering, not the input. In Figures~\ref{fig:HurstFX} and~\ref{fig:Filtre}, we have input an $H$ equal to the large-scale perceived Hurst exponent before denoising\footnote{ We thus take into account the fact that the $H$ perceived at large scales is less affected by noise than the $H$ at small scales.}, as displayed in Table~\ref{tab:HurstFX}, for instance $H=0.244$ for EUR/USD. But whatever the value of $H$, the filtered log-log plot (red) shows a stronger convexity than the grey curve, which is also convex. The lower the input $H$ in the function $f$, the stronger this convexity. If the true dynamic was an fBm of Hurst exponent $H$, the filtered log-log plot, with $H$ as input of $f$, should be a straight line of slope $2H$. Since no $H$ results in such a straight line, the fBm cannot describe on its own the dynamic of log-volatilities.

\end{itemize}

\begin{Remark}
When the time step used in the realized variance increases, the measurement noise increases as well. As a consequence,  the slope of the log-log plot decreases (because of the greater impact of the noise in the estimation of the Hurst exponent), as one can see in Figure~\ref{fig:NoiseStepRV_EURUSD}, but empirically the global shape of the log-log plot remains unchanged, still showing convexity. In other words, the choice of the number of time steps in the computation of the realized variance has empirically a limited impact on the estimation of the Hurst exponent in our dataset.
\begin{figure}[htbp]
	\centering
		\includegraphics[width=0.7\textwidth]{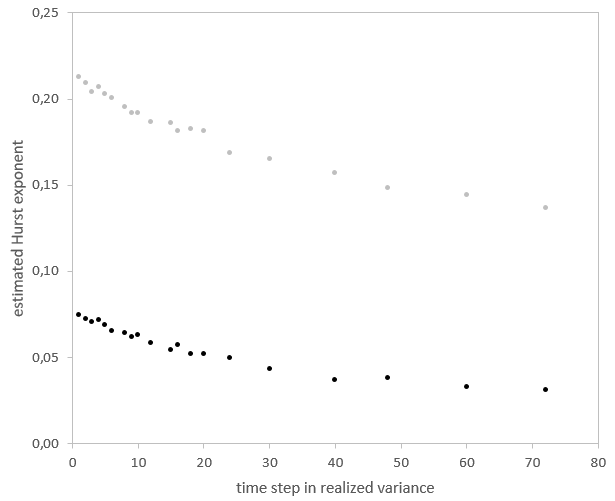}
\begin{minipage}{0.7\textwidth}\caption{Estimated Hurst exponent at small scales (lower than 21 days, in black) and high scales (between 54 and 149 days, in grey) for various time steps in the realized variance, obtained on EUR/USD.}
	\label{fig:NoiseStepRV_EURUSD}
\end{minipage}
\end{figure}
\end{Remark}

\section{Conclusion}\label{Conclusion}

In this paper we performed a extensive empirical investigation on a large dataset of exchange rates. Our findings are twofold: first, we confirmed that the estimation of the Hurst exponent is indeed below 0.5 for time scales close to the ones considered in past literature, e.g. in \cite{GJR}. However, when larger  time scales are considered (which is possible thanks to the size of our dataset), we face a violation of the stationarity assumption of the increments of the  fBm  driving the volatility process, in contrast with the rough volatility paradigm. In fact, in the log-log plot we observe a convexity effect that cannot be explained even by taking into account the different types of noises in the estimation procedures. As a consequence, we consider this convexity effect as a new stylized fact of the market that any  advanced volatility model should be able to reproduce.

\bibliographystyle{plainnat}
\bibliography{bibliorough}

\end{document}